\documentclass{amsart}
\usepackage[ruled,vlined,linesnumbered]{algorithm2e}
\usepackage{amsmath}
\usepackage{amssymb}
\usepackage{tikz}
\usepackage{stmaryrd}

\newtheorem{thm}{Theorem}
\newtheorem{prop}[thm]{Proposition}
\newtheorem{lem}[thm]{Lemma}
\theoremstyle{remark}
\newtheorem*{rem}{Remark}

\newcommand{\assign}{\leftarrow}
\newcommand{\ZZ}{\mathbf{Z}}
\newcommand{\CC}{\mathbf{C}}
\newcommand{\CCxx}{\CC \llbracket x \rrbracket}
\newcommand{\fourier}{\mathcal{F}}

\SetKwInOut{Output}{Output}
\SetKw{KwAnd}{and}

\begin{document}

\title[Faster algorithms for the square root and reciprocal]{Faster algorithms for the square root and reciprocal of power series}
\author{David Harvey}

\begin{abstract}
We give new algorithms for the computation of square roots and reciprocals of power series in $\CCxx$. If $M(n)$ denotes the cost of multiplying polynomials of degree $n$, the square root to order $n$ costs $(1.333\ldots + o(1)) M(n)$ and the reciprocal costs $(1.444\ldots + o(1)) M(n)$. These improve on the previous best results, respectively $(1.8333\ldots + o(1)) M(n)$ and $(1.5 + o(1)) M(n)$.
\end{abstract}

\maketitle

\section{Introduction}
\label{sec:intro}

It has been known for some time that various operations on power series, such as division, reciprocal, square root, and exponentiation, may be performed in a constant multiple of the time required for a polynomial multiplication of the same length. More recently the focus has been on improving the constants. A wealth of historical and bibliographical information tracking the downward progress of these constants may be found in \cite{bernstein} and \cite{multapps}. In this paper we present results that improve on the best known constants for the \emph{square root} and \emph{reciprocal} operations.

Let $M(n)$ denote the cost of multiplying two polynomials in $\CC[x]$ of degree less than $n$. By `cost' or `running time' we always mean number of ring operations in $\CC$. We assume FFT-based multiplication throughout, so that $M(n) = O(n \log n)$.

Let $f \in \CCxx$, $f = 1 \bmod x$. There are two algorithms for computing $f^{-1} \bmod x^n$ that achieve the previously best known running time bound of $(1.5 + o(1)) M(n)$. The first is that of Sch\"onhage \cite[Theorem 2]{schonhage}. If $g$ is an approximation to $f^{-1}$ of length $k$, then the second-order Newton iteration $g' = (2g - g^2(f \bmod x^{2k}))$ yields an approximation to $f^{-1}$ of length $2k$. Sch\"onhage observed that it suffices to compute $g^2 (f \bmod x^{2k})$ modulo $x^{3k} - 1$, which can be achieved by two forward FFTs and one inverse FFT of length $3k$. Iterating this process, the running time bound follows easily. Bernstein's `messy' algorithm \cite[p.~10]{bernstein} is more complicated. Roughly speaking, he splits the input into blocks of consecutive coefficients, and applies a second-order Newton iteration at the level of blocks. This blocking strategy allows transforms of blocks to be reused between iterations.

Our new reciprocal algorithm may be viewed as a simultaneous generalization of Bernstein's reciprocal algorithm and van der Hoeven's division algorithm \cite[p.~6]{vdh}. The main innovation is the use of a third-order Newton iteration, whose additional term is computed essentially free of charge, leading to a running time of $(1.444\ldots + o(1)) M(n)$ (Theorem \ref{thm:recip}). Although this is only a small improvement, it is interesting theoretically because the `nice' bound $(1.5\ldots + o(1)) M(n)$, achieved by two quite different algorithms, had been a plausible candidate for the optimal bound for almost ten years. Furthermore, the methods presented in this paper suggest that $(1.333\ldots + o(1)) M(n)$ may be attainable (see the final remark in Section \ref{sec:recip}).

For the square root, there are again two contenders for the previously best known bound of $(1.8333\ldots + o(1)) M(n)$. Bernstein computes the square root and reciprocal square root together, alternately extending approximations of each \cite[p.~9]{bernstein}. Hanrot and Zimmermann first compute the reciprocal square root to half the target precision, using a technique similar to Sch\"onhage's, and then apply a different iteration at the last step to obtain the square root \cite{hz}. (They claim only 1.91666\ldots, but there is an error in their analysis; the cost of line 3 of Algorithm `SquareRoot' is $M(n)/3$, not $M(n)/2$.)

Our new square root algorithm achieves $(1.333\ldots + o(1)) M(n)$ (Theorem \ref{thm:sqrt}). It is quite different to both of the above algorithms. It operates on blocks of coefficients, and may be viewed as a straightforward adaptation of van der Hoeven's division algorithm to the case of extracting square roots.

For simplicity, in this paper we only discuss the case of $\CCxx$. It seems likely that the algorithms may also be adapted to achieve the same constants in the case of power series over an arbitrary ring, and also in the case of arbitrary-precision integers or real numbers, but we have not checked the details.

\section{Notation and complexity assumptions}

The Fourier transform $\fourier_n(g) \in \CC^n$ of a polynomial $g \in \CC[x]$ is defined by $(\fourier_n(g))_j = g(e^{2\pi ij/n})$. If $\deg g < n$, we denote by $T(n)$ the cost of computing $\fourier_n(g)$ from $g$, or of computing $g$ from $\fourier_n(g)$.

If $g_1, g_2 \in \CC[x]$ and $\deg g_i < n$, the cyclic convolution $g_1 g_2 \bmod x^n - 1$ may be computed by evaluating $\fourier_n^{-1} (\fourier_n(g_1) \fourier_n(g_2))$, where the Fourier transforms are multiplied componentwise. The running time is $3T(n) + O(n)$. To obtain the ordinary product $g_1 g_2$ one may compute $g_1 g_2 \bmod x^{2n'} - 1$ for any $n' \geq n$, leading to the estimate $M(n) = 3T(2n') + O(n')$. While it is known that $T(n) = O(n \log n)$ for all $n$, for sufficiently smooth $n$ the implied big-$O$ constant may be smaller than the worst case, and one should choose $n'$ to take advantage of this. We therefore assume that we have available a set $S \subseteq \ZZ^+$ with the following properties: first, that $T(2n) = (1/3 + o(1)) M(n)$ for $n \in S$, and second, that the ratio of successive elements of $S$ approaches $1$. The choice of $S$ will depend on exactly which FFT algorithms are under consideration. For example, Bernstein describes a particular class of FFT algorithms for which the above properties hold with $S = \{2^k m: \text{$m$ odd and $k \geq m^2 - 1$}\}$ \cite[p.~5]{bernstein}. Following Bernstein, we call elements of $S$ \emph{ultrasmooth} integers.

If $g, h \in \CC[x]$, $\deg g < 2n$, $\deg h < n$, we denote by $g \rtimes_n h$ the middle product of $g$ and $h$. That is, if $gh = p_0 + p_1 x^n + p_2 x^{2n}$ where $p_i \in \CC[x]$, $\deg p_i < n$, then by definition $g \rtimes_n h = p_1$. Note that $gh = (p_0 + p_2) + p_1 x^n \bmod x^{2n} - 1$, so $g \rtimes_n h$ may be computed by evaluating $\fourier_{2n}^{-1} (\fourier_{2n}(g) \fourier_{2n}(h))$ and discarding the first half of the output. See \cite{tellegen} and \cite{hqz} for more information about the middle product.

In the algorithms given below, we fix a block size $m \geq 1$, and for $f \in \CCxx$, we write $f = f_{[0]} + f_{[1]} X + f_{[2]} X^2 + \cdots$, where $X = x^m$ and $\deg f_{[i]} < m$.

\section{Blockwise multiplication of power series}

Our main tool is a technique for multiplying power series, described in the proof of Lemma \ref{lem:block-product} below. Bernstein used a similar idea in \cite[p.~10]{bernstein}. We follow van der Hoeven's more recent approach \cite{vdh}, which uses the middle product to obtain a neater algorithm.

\begin{lem}
\label{lem:block-product}
Let $f, g \in \CCxx$ and $k \geq 0$. Given as input $\fourier_{2m}(f_{[0]}), \ldots, \fourier_{2m}(f_{[k]})$ and $\fourier_{2m}(g_{[0]}), \ldots, \fourier_{2m}(g_{[k]})$, we may compute $(fg)_{[k]}$ in time $T(2m) + O(m(k+1))$.
\end{lem}
\begin{proof}
As shown in Figure \ref{fig:product}, we have
 \[ (fg)_{[k]} = \sum_{i=0}^k (f_{[k-i-1]} + f_{[k-i]} X) \rtimes_m g_{[i]}, \]
where for convenience we declare that $f_{[-1]} = 0$. Therefore $(fg)_{[k]}$ is obtained as the second half of
 \[ \fourier_{2m}^{-1}\left( \sum_{i=0}^k  \left(\fourier_{2m}(f_{[k-i-1]}) + \fourier_{2m}(f_{[k-i]})\fourier_{2m}(X)\right) \fourier_{2m}(g_{[i]})\right). \]
Since $(\fourier_{2m}(X))_j = (-1)^j$, the above expression may be computed from the $\fourier_{2m}(f_{[i]})$ and $\fourier_{2m}(g_{[i]})$ using $O(m(k+1))$ ring operations, followed by a single inverse transform of length $2m$.
\end{proof}

\begin{rem}
In Section \ref{sec:recip}, we will also need to compute expressions of the form $(fg + f'g')_{[k]}$, using the transforms of the blocks of $f$, $f'$, $g$ and $g'$ as input. Since the Fourier transform is linear, the same running time bound $T(2m) + O(m(k+1))$ applies (with a larger big-$O$ constant).
\end{rem}

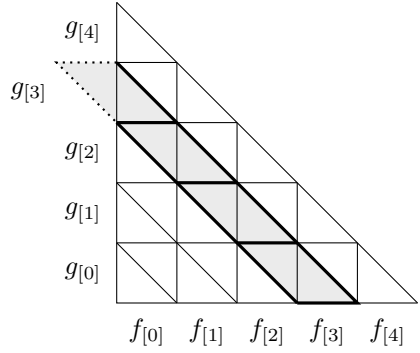
\begin{figure}
\begin{tikzpicture}[scale=0.8]

\fill[fill=black!8] (3,0) -- (4,0) -- (0,4) -- (-1,4) -- (3,0);

\draw (5,0) -- (0,0) -- (0,5);
\draw (0,5) -- (5,0);
\draw (0,4) -- (4,0);
\draw (0,3) -- (3,0);
\draw (0,2) -- (2,0);
\draw (0,1) -- (1,0);

\draw (0,1) -- (4,1);
\draw (0,2) -- (3,2);
\draw (0,3) -- (2,3);
\draw (0,4) -- (1,4);

\draw (1,0) -- (1,4);
\draw (2,0) -- (2,3);
\draw (3,0) -- (3,2);
\draw (4,0) -- (4,1);

\draw[very thick] (3,0) -- (4,0) -- (0,4);
\draw[thick,dotted] (0,4) -- (-1,4) -- (0,3);
\draw[very thick] (0,3) -- (3,0);
\draw[very thick] (2,1) -- (3,1);
\draw[very thick] (1,2) -- (2,2);
\draw[very thick] (0,3) -- (1,3);

\draw (0.5,-0.1) node[below] {$f_{[0]}$};
\draw (1.5,-0.1) node[below] {$f_{[1]}$};
\draw (2.5,-0.1) node[below] {$f_{[2]}$};
\draw (3.5,-0.1) node[below] {$f_{[3]}$};
\draw (4.5,-0.1) node[below] {$f_{[4]}$};

\draw (-0.1, 0.5) node[left] {$g_{[0]}$};
\draw (-0.1, 1.5) node[left] {$g_{[1]}$};
\draw (-0.1, 2.5) node[left] {$g_{[2]}$};
\draw (-1.0, 3.5) node[left] {$g_{[3]}$};
\draw (-0.1, 4.5) node[left] {$g_{[4]}$};

\end{tikzpicture}
\caption{Blockwise product of power series. Terms contributing to $(fg)_{[3]}$ are shaded.}
\label{fig:product}
\end{figure}

\section{Square root}
\label{sec:sqrt}

If $f = f_0 + f_1 x + f_2 x^2 + \cdots$ and $g = f^{1/2} = g_0 + g_1 x + g_2 x^2 + \cdots$, then the coefficients of $g$ may be determined by solving sequentially $g_0^2 = f_0$, $2 g_0 g_1 = f_1$, $2 g_0 g_2 = f_2 - g_1^2$, $2 g_0 g_3 = f_3 - 2 g_1 g_2$, $2 g_0 g_4 = f_4 - (2 g_1 g_3 + g_2^2)$, and so on. Algorithm \ref{algo:sqrt} applies this procedure at the level of blocks, retaining the Fourier transform of each computed block as it proceeds.

\begin{algorithm}
\label{algo:sqrt}
\dontprintsemicolon
\KwIn{$r \in \ZZ$, $r \geq 1$ \newline
$f \in \CCxx$, $f = 1 \bmod x$ \newline
$g_{[0]} = (f_{[0]})^{1/2} \bmod X$ \newline
$h = (g_{[0]})^{-1} \bmod X$
}\;
\KwOut{$g = g_{[0]} + \cdots + g_{[r-1]} X^{r-1} = f^{1/2} \bmod X^r$}\;
\medskip
Compute $\fourier_{2m}(h)$\;
\For{$1 \leq k < r$}{
   Compute $\fourier_{2m}(g_{[k-1]})$\;
   $\psi \assign ((g_{[0]} + \cdots + g_{[k-1]} X^{k-1})^2)_{[k]}$\;
   Compute $\fourier_{2m}(f_{[k]} - \psi)$\;
   $g_{[k]} \assign \frac12 h (f_{[k]} - \psi) \bmod X$\;
}\;
\caption{Square root}
\end{algorithm}

\begin{prop}
\label{prop:sqrt}
Algorithm \ref{algo:sqrt} is correct, and runs in time $(4r-3) T(2m) + O(r^2 m)$.
\end{prop}
\begin{proof}
By definition $g_{[0]}$ is correct. In the $k$th iteration of the loop, assume that $g_{[0]}, \ldots, g_{[k-1]}$ have been computed correctly. Then we have
\[ \begin{aligned}
 (g_{[0]} + \cdots + g_{[k-1]} X^{k-1})^2 & = f_{[0]} + \cdots + f_{[k-1]} X^{k-1} + \psi X^k & & \bmod X^{k+1}, \\
 (g_{[0]} + \cdots + g_{[k]} X^k)^2 & = f_{[0]} + \cdots + f_{[k-1]} X^{k-1} + f_{[k]} X^k & & \bmod X^{k+1}.
\end{aligned} \]
Subtracting these yields $2 g_{[0]} g_{[k]} = f_{[k]} - \psi \bmod X$, so $g_{[k]}$ is computed correctly.

Each iteration performs one inverse transform to obtain $\psi$ (Lemma \ref{lem:block-product}), one to obtain $g_{[k]}$, and the two forward transforms explicitly stated. The total number of transforms is therefore $4(r-1) + 1 = 4r - 3$. The loop also performs $O(km)$ scalar operations in the $k$th iteration (Lemma \ref{lem:block-product}).
\end{proof}

\begin{thm}
\label{thm:sqrt}
The square root of a power series $f = 1 + f_1 x + \cdots \in \CCxx$ may be computed to order $n$ in time $(4/3 + o(1)) M(n)$.
\end{thm}
\begin{proof}
Let $r \geq 1$, and let $m$ be the smallest ultrasmooth integer larger than $n/r$. We let $r$ grow slowly with respect to $n$; specifically we assume that $r \to \infty$ and $r = o(\log n)$ as $n \to \infty$. Then $m \to \infty$ as $n \to \infty$, so $m = (1 + o(1))n/r$. Zero-pad $f$ up to length $rm$. Compute $g_{[0]} = (f_{[0]})^{1/2} \bmod x^m$ and $h = (g_{[0]})^{-1} \bmod x^m$ using any $O(M(m))$ algorithm. Compute $f^{1/2} \bmod x^{rm}$, hence $f^{1/2} \bmod x^n$, using Algorithm \ref{algo:sqrt}. By Proposition \ref{prop:sqrt} the total cost is
\begin{align*}
 O(M(m)) + (4r-3) T(2m) + O(r^2 m)
  & = (4r/3 + O(1)) M(m) + O(r^2 m) \\
  & = (4/3 + O(r^{-1})) M(mr) + O(r^2 m) \\
  & = (4/3 + O(r^{-1})) M(n) + O(rn) \\
  & = (4/3 + o(1)) M(n),
\end{align*}
assuming that $M(n) = \Theta(n \log n)$.
\end{proof}

\begin{rem}
If $g_{[0]}$ and $h$ are computed using Bernstein's $(2.5 + o(1)) M(n)$ algorithm for the simultaneous computation of the square root and reciprocal square root \cite[p.~9]{bernstein}, then already for $r = 3$ the new algorithm matches the previous bound of $(1.8333\ldots + o(1)) M(n)$, and is strictly faster for $r \geq 4$.
\end{rem}

\begin{rem}
Let $f \in \CC[x]$ be a monic polynomial of degree $2n$. The above algorithm may be adapted to compute the square root with remainder, that is, polynomials $g, h \in \CC[x]$ with $\deg g = n$, $\deg h < n$, and $f = g^2 + h$, in time $(5/3 + o(1)) M(n)$.

For this, write $\tilde f(x) = x^{2n} f(1/x)$, $\tilde g(x) = x^n g(1/x)$, $\tilde h(x) = x^n h(1/x)$. Then $\tilde f, \tilde g, \tilde h \in \CCxx$, and we want to solve $\tilde f(x) = \tilde g(x)^2 + x^n \tilde h(x)$. First compute $\tilde g(x)$ using the above algorithm; to find $\tilde h(x)$ it then suffices to compute $\tilde g(x)^2$. Observe that at the end of Algorithm \ref{algo:sqrt}, we may compute $((\tilde g_{[0]} + \cdots + \tilde g_{[r-1]} X^{r-1})^2)_{[j]}$ for $r \leq j < 2r$ in time $r T(2m) + O(r^2 m)$ using Lemma \ref{lem:block-product}, since the transforms of the $\tilde g_{[j]}$ are all known. This increases the cost from $(4r - 3)T(2m) + O(r^2 m)$ to $(5r - 3)T(2m) + O(r^2 m)$, leading to the claimed bound in the same way as in the proof of Theorem \ref{thm:sqrt}.
\end{rem}

\section{Reciprocal}
\label{sec:recip}

Let $f = 1 + f_1 x + \cdots \in \CCxx$, and suppose that $g = f^{-1} \bmod x^n$. Then $fg = 1 + \delta x^n \bmod x^{3n}$ for some $\delta \in \CC[x]$, $\deg \delta < 2n$. Putting $g' = g(1 - \delta x^n + \delta^2 x^{2n})$, we have $g' = f^{-1} \bmod x^{3n}$. This is the third-order Newton iteration for the reciprocal.

The idea of Algorithm \ref{algo:recip} below is to use the above recipe at the level of blocks, with an additional twist. If we write $\delta = \delta_0 + \delta_1 x^n$, where $\deg \delta_i < n$, then $g' = g(1 - \delta_0 x^n + (\delta_0^2 - \delta_1) x^{2n}) \bmod x^{3n}$. The algorithm first computes $\delta_0$, applying Lemma \ref{lem:block-product} to compute the relevant blocks of $fg$. Then, instead of computing $\delta_0^2$ and $\delta_1$ separately, it computes the sum $\delta_0^2 - \delta_1$ in one pass, using only one inverse transform per block (see the remark after Lemma \ref{lem:block-product}). This is possible since $\delta_0$ is already completely known, and constitutes the main source of savings over Bernstein's algorithm.

\begin{algorithm}
\label{algo:recip}
\dontprintsemicolon
\KwIn{$s \in \ZZ$, $s \geq 1$ \newline
$f \in \CCxx$, $f = 1 \bmod x$ \newline
$g_{[0]} = (f_{[0]})^{-1} \bmod X$
}\;
\KwOut{$g = g_{[0]} + \cdots + g_{[3s-1]} X^{3s-1} = f^{-1} \bmod X^{3s}$}\;
\medskip\;
Compute $\fourier_{2m}(g_{[0]})$\;
\lFor{$0 \leq i < 3s$}{compute $\fourier_{2m}(f_{[i]})$}\;
\For{$1 \leq k < s$}{\nllabel{line:div-1}
   $\psi \assign ((f_{[0]} + \cdots + f_{[k]} X^k)(g_{[0]} + \cdots + g_{[k-1]} X^{k-1}))_{[k]}$\;
   Compute $\fourier_{2m}(\psi)$\;
   $g_{[k]} \assign - g_{[0]} \psi \bmod X$\;
   Compute $\fourier_{2m}(g_{[k]})$\nllabel{line:div-2}\;
}\;
\For{$0 \leq k < s$}{\nllabel{line:low-d-1}
$d_{[k]} \assign -((f_{[0]} + \cdots + f_{[3s-1]} X^{3s-1})(g_{[0]} + \cdots + g_{[s-1]} X^{s-1}))_{[k+s]}$\;
Compute $\fourier_{2m}(d_{[k]})$\nllabel{line:low-d-2}
}\;
\For{$s \leq k < 2s$}{$d_{[k]} \assign ((d_{[0]} + \cdots + d_{[s-1]} X^{s-1})^2)_{[k-s]}$\;
$\hspace{35pt} \mathord{} \mathbin{-} ((f_{[0]} + \cdots + f_{[3s-1]} X^{3s-1})(g_{[0]} + \cdots + g_{[s-1]} X^{s-1}))_{[k+s]}$\;
Compute $\fourier_{2m}(d_{[k]})$}\;
\For{$s \leq k < 3s$}{$g_{[k]} \assign ((d_{[0]} + \cdots + d_{[2s-1]} X^{2s-1})(g_{[0]} + \cdots + g_{[s-1]} X^{s-1}))_{[k-s]}$}\;
\caption{Reciprocal}
\end{algorithm}

\begin{prop}
\label{prop:recip}
Algorithm \ref{algo:recip} is correct, and runs in time $(13s - 3) T(2m) + O(s^2 m)$.
\end{prop}
\begin{proof}
By definition $g_{[0]}$ is correct. Lines \ref{line:div-1}--\ref{line:div-2} compute $g_{[1]}, \ldots, g_{[s-1]}$ as follows. In the $k$th iteration of the loop, assume that $g_{[0]}, \ldots, g_{[k-1]}$ are correct. Then
\[ \begin{aligned}
 (f_{[0]} + \cdots + f_{[k]} X^k)(g_{[0]} + \cdots + g_{[k-1]} X^{k-1}) & = 1 + \psi X^k & & \bmod X^{k+1}, \\
 (f_{[0]} + \cdots + f_{[k]} X^k)(g_{[0]} + \cdots + g_{[k]} X^k) & = 1 & & \bmod X^{k+1}.
\end{aligned} \]
Subtracting yields $f_{[0]} g_{[k]} = -\psi \bmod X$, so $g_{[k]}$ is computed correctly. (This loop is essentially van der Hoeven's division algorithm, applied to compute $1/f \bmod X^s$.)

Now we use the symbols $\delta$, $\delta_0$, $\delta_1$ introduced earlier, putting $n = sm$. After lines \ref{line:low-d-1}--\ref{line:low-d-2} we have $d_{[0]} + \cdots + d_{[s-1]} X^{s-1} = -\delta_0$, and the subsequent loop computes $d_{[s]} + \cdots + d_{[2s-1]} X^{s-1} = \delta_0^2 - \delta_1 \bmod X^s$. Therefore $d_{[0]} + \cdots + d_{[2s-1]} X^{2s-1} = - \delta + \delta^2 X^s \bmod X^{2s}$. The final loop computes the appropriate blocks of $g' = g(1 - \delta X^s + \delta^2 X^{2s}) \bmod X^{3s}$.

Altogether the algorithm performs $1 + 3s + 2(s-1) + s + s$ forward transforms, $2(s-1) + s + s + 2s$ inverse transforms, and $O(s^2 m)$ scalar operations (apply Lemma \ref{lem:block-product} to each loop).
\end{proof}

\begin{thm}
\label{thm:recip}
The reciprocal of a power series $f \in \CCxx$ may be computed to order $n$ in time $(13/9 + o(1)) M(n)$.
\end{thm}
\begin{proof}
Apply the proof of Theorem \ref{thm:sqrt} to Proposition \ref{prop:recip}, with $r = 3s$.
\end{proof}

\begin{rem}
If $g_{[0]}$ is computed using a $(1.5 + o(1)) M(n)$ algorithm, the new algorithm achieves the same bound for $s = 3$ ($r = 9$), and is faster for $s \geq 4$ ($r \geq 12$).
\end{rem}

\begin{rem}
A natural question is whether the key idea of Algorithm \ref{algo:recip} can be extended to Newton iterations of arbitrarily high order. That is, if $fg = 1 + \delta x^n$, is it possible to compute $1 - \delta x^n + \delta^2 x^{2n} \cdots \pm \delta^{k-1} x^{(k-1)n} \bmod x^{kn}$ in essentially the same time as $1 + \delta x^n \bmod x^{kn}$ itself, for arbitrary $k$? Algorithm \ref{algo:recip} corresponds to the case $k = 3$. An affirmative answer for arbitrary $k$ would presumably lead to a $(1.333\ldots + o(1)) M(n)$ algorithm for the reciprocal.
\end{rem}

\section*{Acknowledgments}

Many thanks to Paul Zimmermann for his suggestions that greatly improved the presentation of these results.

\bibliographystyle{amsalpha}
\bibliography{fast-series}

\end{document}